\documentclass[11pt]{article}
\usepackage[letterpaper,hmargin=1in,vmargin=1in]{geometry}
\usepackage{enumitem}
\usepackage{graphicx}				
\usepackage{amsmath}				
\usepackage{amsthm}					
\usepackage{amssymb}				
\usepackage{amsfonts} 				
\usepackage{thmtools}				
\usepackage{url}					
\usepackage{cite}					
\usepackage{hyperref}				
\hypersetup{colorlinks=true,linkcolor=[rgb]{0.75,0,0},citecolor=[rgb]{0,0,0.75}}%
\declaretheorem[name=Theorem]{theorem}
\declaretheorem[name=Lemma,numberwithin=section]{lemma}

\newcommand{\ang}[1]{\langle #1\rangle} 
\newcommand{\RE}{\mathbb{R}}            
\newcommand{\eps}{\varepsilon}          
\newcommand{\SP}{\kern+1pt}				
\newcommand{\SSP}{\kern+2pt}			

\DeclareMathOperator*{\OPT}{OPT}
\newcommand{\mtsp}{m$^3$TSP}
\usepackage{xcolor}

\begin{document}
\title{A PTAS for the Min-Max Euclidean Multiple TSP}

\author{Mary Monroe \\
        Amazon Web Services \\
        Crystal City, Arlington, Virginia \\
        mfmonroe@terpmail.umd.edu
	\and
	David M. Mount\\
		Dept.\ of Computer Science and \\
		Inst.\ for Advanced Computer Studies \\
		University of Maryland \\
		College Park, Maryland \\
		mount@umd.edu 
}

\date{\today}

\maketitle

\begin{abstract}
We present a polynomial-time approximation scheme (PTAS) for the min-max multiple TSP problem in Euclidean space, where multiple traveling salesmen are tasked with visiting a set of $n$ points and the objective is to minimize the maximum tour length. For an arbitrary $\eps > 0$, our PTAS achieves a $(1 + \eps)$-approximation in time $O \big(n ((1/\eps) \log (n/\eps))^{O(1/\eps)} \big)$. Our approach extends Arora's dynamic-programming (DP) PTAS for the Euclidean TSP~\cite{Arora98}. Our algorithm introduces a rounding process to balance the allocation of path lengths among the multiple salesman. We analyze the accumulation of error in the DP to prove that the solution is a $(1 + \eps)$-approximation.
\end{abstract}

\section{Introduction} \label{sec:intro}

The \emph{min-max multiple traveling salesmen problem} (or \emph{\mtsp}) is a generalization of the traveling salesman problem (TSP). Instead of one traveling salesman, the {\mtsp} involves multiple salesmen who together visit a group of designated cities with the objective of minimizing the maximum distance traveled by any member of the group. We consider the problem in the Euclidean setting, where the cities to be visited are points in $\RE^2$, and the distance between cities is the Euclidean distance. We support a formulation where each salesman begins and ends at the same designated depot. We present a randomized polynomial-time approximation scheme (PTAS) for this problem. 

While the majority of research on the multiple TSP problem has focused on minimizing the sum of travel distances, the min-max formulation is most appropriate for applications where the tours are executed in parallel, and hence the natural objective is to minimize the maximum service time (or \emph{makespan}). We assume that all salesmen travel at the same speed, and thus the objective is to minimize the maximum distance traveled by any of the salesmen.

\subsection{Definitions and Results} \label{subsec:defs}

The input to {\mtsp} consists of a set $C = \{c_1, \ldots, c_n\}$ of $n$ points (or ``cities'') in $\RE^2$, the number $k$ of salesmen, and a depot $d \in C$. The output is a set of $k$ \emph{tours}, where the $h$th tour is a sequence of cities that starts and ends at the depot. All of the cities of $C$ must be visited by at least one salesman. The length of a tour is just the sum of inter-city distances under the Euclidean metric, and the objective is to minimize the maximum tour length over all $k$ tours. In the approximate version, we are also given an approximation parameter $\eps > 0$, and the algorithm returns a set of $k$ tours such that the maximum tour length is within a factor of $(1 + \eps)$ of the optimal min-max length. Throughout, we assume that $k$ is a fixed constant, and we treat $n$ and $\eps$ as asymptotic quantities. Here is our main result.

\begin{theorem} \label{thm:main}
Given a set of $n$ points in $\RE^2$, a fixed integer constant $k \geq 1$, and a approximation parameter $\eps > 0$, there exists a randomized algorithm for {\mtsp} that runs in expected time $O\big(n ((1/\eps) \log (n/\eps))^{O(1/\eps)} \big)$.
\end{theorem}

Note that the big-O notation in the polynomial's exponent conceals a constant factor that depends on $k$.

\subsection{Related Work} \label{subsec:related_work}

The multiple traveling salesmen problem with the min-max objective has been studied from both a theoretical and practical perspective. Arkin, Hassin, and Levin describe approximation algorithms for a variety of vehicle routing problems within graphs~\cite{Ark06}. This includes the minimum path cover, where we are given a bound on each path length and the objective is to minimize the number of salesmen required to visit each node. Most notably, they provide a 3-factor approximation to the min-max path cover, which is equivalent to our {\mtsp} problem (in graphs, not Euclidean space) when depot locations are assigned to each salesman.  Xu and Rodrigues give a $3/2$-factor approximation for the min-max multiple TSP in graphs where there are multiple distinct depots~\cite{Xu10}. In both cases, the number of salesmen $k$ is taken to be a constant.

Researchers have also explored different heuristic approaches to solving the min-max problem. França, Gendreau, Laporte, and Müller provide two exact algorithms for {\mtsp} in graphs with a single depot. They also present an empirical analysis of their algorithms with cases for $\leq 50$ cities and $m = 2,3,4,5$ ~\cite{Franc95}. Other approaches include neural networks and tabu search (see, e.g., \cite{Mats14, Franc95}). Carlsson, \textit{et al.} apply a combination of linear programming with global improvement and a region partitioning heuristic~\cite{Car07}.

Even, Garg, Könemann, Ravi, and Sinha present constant factor approximations for covering the nodes of a graph using trees with the min-max objective~\cite{Even04}. The algorithms run in polynomial time over the size of the graph and $\log({1/\eps})$. In this case, the trees are either rooted or unrooted, and there is a given upper bound on the number of trees involved. The authors describe the instance as the ``nurse station location'' problem, where a hospital needs to send multiple nurses to different sets of patients such that the latest completion time is minimized. 

Becker and Paul considered multiple TSP in the context of navigation in trees~\cite{Be19}. Their objective is the same as ours, but distances are measured as path lengths in a tree. They present a PTAS for the min-max routing problem in this context. In their formulation, the root of the tree serves as the common depot for all the tours. They strategically split the tree into clusters that are covered by only one or two salesmen. Their algorithm is based on dynamic programming, where each vertex contains a limited configuration array of rounded values representing the possible tour lengths up to that point. This rounding reduces the search space and allows for a polynomial run-time. Our approach also applies a form of length rounding, though our strategy differs in that it rounds on a logarithmic scale and performs dynamic-programming on a quadtree structure.

Asano, Katoh, Tamaki, and Tokuyama considered a multi-tour variant of TSP, where for a given $k$ the objective is to cover a set of points by tours each of which starts and ends at a given depot and visits at most $k$ points of the set~\cite{Asano97}. They present a dynamic programming algorithm in the style of Arora's PTAS whose running time is $(k/\eps)^{O(k/\eps^3)} + O(n \log n)$.

\subsection{Classic TSP PTAS}

In his landmark paper, Arora presented a PTAS for the standard TSP problem in the Euclidean metric~\cite{Arora98}. Since we will adopt a similar approach, we provide a brief overview of his algorithm. It begins by rounding the coordinates of the cities so they lie on points of an integer grid, such that each pair of cities is separated by a distance of at least eight units. (See our version of perturbation in Section~\ref{subsec:perturbation}.) This step introduces an error of at most $\frac{\eps \OPT}{4}$ to the tour length, where $\OPT$ is the length of the optimum tour. Let $L$ denote the side length of the bounding box after perturbation; note that $L = O(n/\eps)$.

The box is first partitioned through a quadtree dissection, in which a square cell is subdivided into four child squares, each of half the side length. This continues until a cell contains at most one point. Due to rounding, each cell has side length at least $1$  (see Figure~\ref{fig:quadtree}, left). The dissection is randomized by introducing an $(a, b)$-shift. The values $a$ and $b$ are randomly-chosen integers in $[0, L)$; an $(a,b)$ shift adds $a$ to the $x$ coordinate and $b$ to the $y$ coordinate of the dissection's lines before reducing the value by modulo $L$ (see Figure~\ref{fig:quadtree}, right). We identify each square in the quadtree based on its level in the recursion, so that a level $i$ node is a square created after performing our recursive split $i$ times.

\begin{figure}[htbp]
  \centerline{\includegraphics[scale=1]{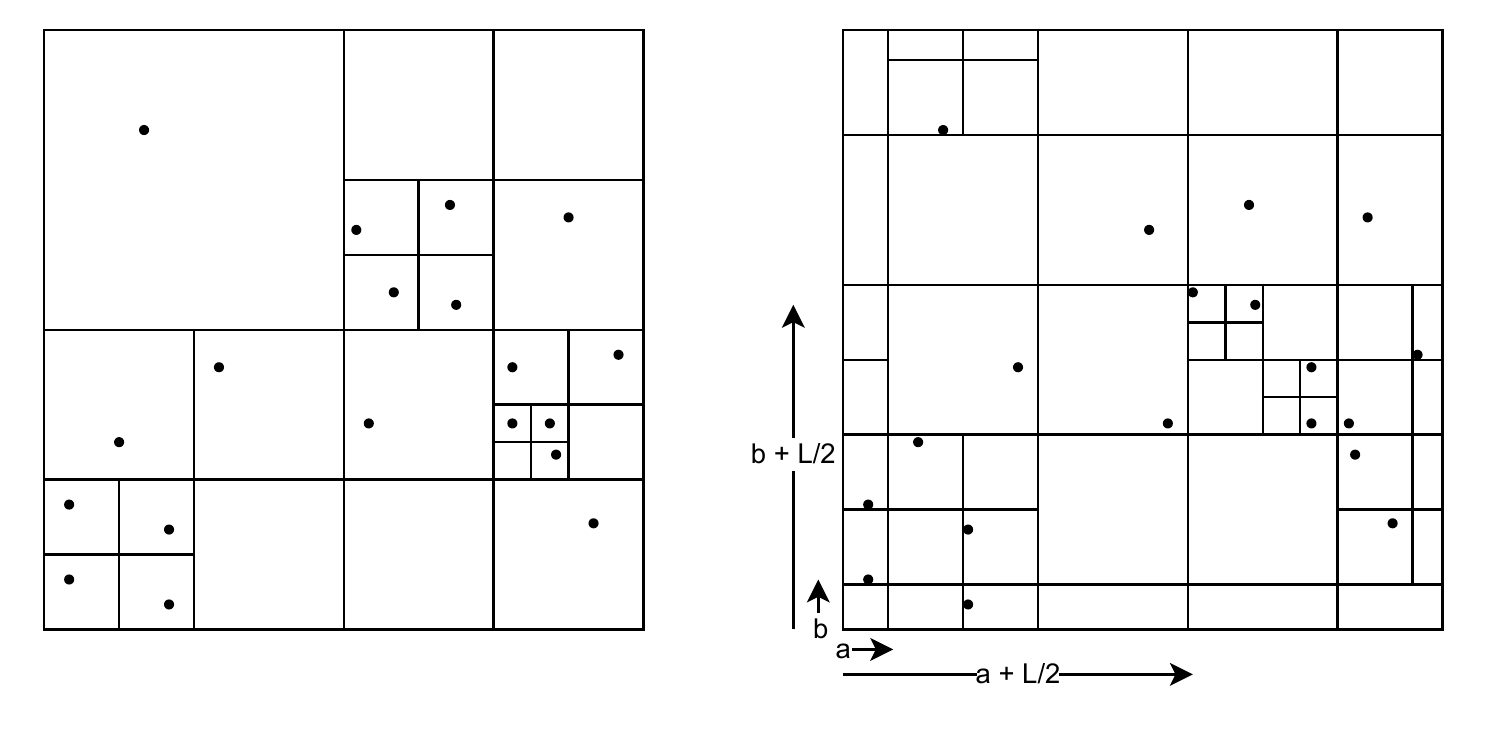}}
  \caption{An example of a quadtree structure incurred on a set of points. The left shows the original quadtree structure, while the right displays an $(a,b)$-shift applied to the same set of points.}
  \label{fig:quadtree}
\end{figure}

The meat of Arora's algorithm lies in his notion of portals. He defines an $m$-regular set of portals for the shifted dissection as a set of $m$ evenly spaced points, called \emph{portals}, placed on each of the four edges of a quadtree square. Moreover, one portal is placed on each corner. He defines a salesman tour as \emph{$(m,r)$-light} if it crosses each edge of the square only at a portal, at most $r$ times. His DP algorithm proceeds bottom-up from the leaves of the quadtree, calculating the optimal TSP path within the square subject to the $(m,r)$-light restriction. In the end, this approach takes polynomial time over $n$.

Arora proves in his \emph{Structure Theorem} that, given a randomly-chosen $(a,b)$-shift, with probability at least $1/2$ the $(m,r)$-light solution has a cost of at most $(1 + \eps) \OPT$ (where $m = O\left(\frac{\log L}{\eps}\right)$ and $r = O\left(\frac{1}{\eps}\right)$). The proof shows that perturbing the optimal tour to be $(m,r)$-light with respect to each quadtree square adds only a small amount of error. The probability of cost is calculated through the expected value of extra length charged to each line in the quadtree.

Our algorithm shares elements in common with Arora's algorithm, including quadtree-based dissection, random shifting, and $(m,r)$-light portal restrictions, all within a dynamic-programming approach. However, to handle multiple tours, we alter the DP formulation and develop a rounding strategy to balance the lengths of the various tours.

\section{Algorithm} \label{sec:algorithm}

In this section we present our algorithm for {\mtsp}.

\subsection{Perturbation} \label{subsec:perturbation}

Before computing the dissection on which the DP is based, we begin by perturbing and rounding the points to a suitably sized square grid in the same manner as Arora. That is, we require that (i) all points have integral coordinates, (ii) any nonzero distance between two points is at least eight units, and (iii) the maximum distance between two points is $O(n/\eps)$. Take $L_0$ as the original size of the bounding box of the point set, and let $\OPT$ denote the optimal solution's makespan (bearing in mind that $\OPT \geq L_0/k$). To accomplish (i), a grid of granularity $\frac{\eps L_0}{8 k n}$ is placed and all points are moved to the closest coordinate. This means that, for a fixed order of nodes visited, the maximum tour length is increased by at most $2 n \frac{\eps L_0}{8 k n} = \frac{\eps L_0}{4 k} < \frac{\eps \OPT}{4}$. To accomplish (ii) and (iii), all distances are divided by $\frac{\eps L_0}{64 k n}$. This step leads to nonzero internode distances of at least $\frac{\eps L_0}{8 k n} / \frac{\eps L_0}{64 k n} = 8$. Moreover, we have that $L = \frac{64kn}{\eps} = O(n/\eps)$, so the maximum internode distance is $O(n/\eps)$. Note that the error incurred on $\OPT$ by snapping points to the grid could be as much as $\frac{\eps \OPT}{4}$.

\subsection{DP Program Structure} \label{subsec:dp_structure}

We define the same $(m,r)$-light portal restriction for each of our tours; however, in our case we have $m = O\left(\frac{k\log L}{\eps}\right)$ and $r = O\left(\frac{k}{\eps}\right)$ (see \ref{lem:structure_arora} for the specific values). 

On top of Arora's portal-pairing argument for each quadtree node, leading to a DP lookup table of tours, we add another restriction for our $k$ tours. For each separate TSP tour, in each quadtree node, we store all possible path lengths as booleans rounded on a certain scale. We need this rounding to manage our algorithm's running time, for otherwise, the number of possible tour lengths would explode as we climb further up our DP table.

\begin{figure}[htbp]
  \centerline{\includegraphics[scale=0.75]{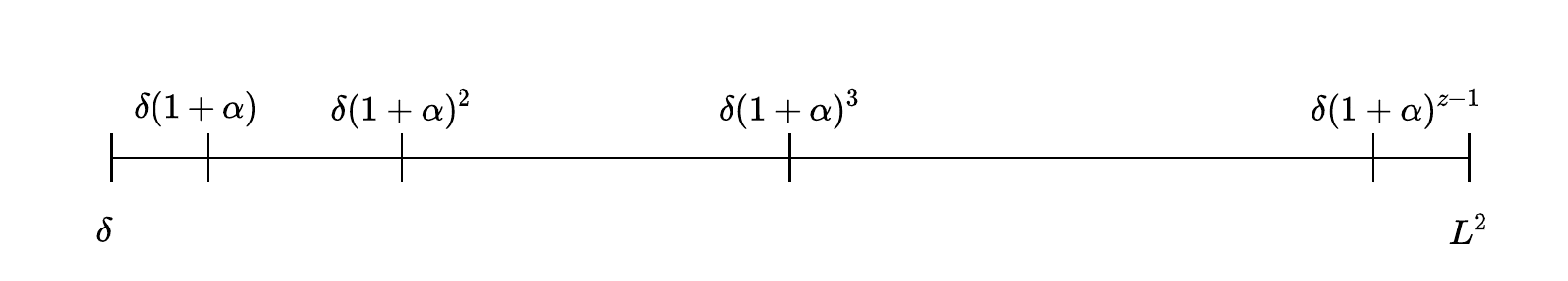}}
  \caption{A visualization of the rounding scale}
  \label{fig:dp_scale}
\end{figure}

Our algorithm rounds any given path length within a level $i$ node to the nearest tick mark above. This scale, shown in Figure 1, increases by a multiplicative factor of $1+\alpha$. Its lower bound, $\delta$, equal to Arora's interportal distance within the level, is $\frac{L}{2^i m}$ (where $L = \frac{64kn}{\eps}$ is our bounding box side length). Arora's single tour TSP solution allows for $\delta$ amount of error in adjusting the tour to only cross edges at portals. Meanwhile, $\alpha = \frac{2 \eps}{\log L} $ represents the degree of granularity we have in our rounding scale (the total error in rounding a tour of length $l$ can be no larger than $(1 + \alpha)l$).

The upper bound on the scale is obtained by applying a basic algorithm to visit each point (which are all snapped to a grid thanks to our perturbation step in section~\ref{subsec:perturbation}). One salesman is chosen that simply zig-zags through each grid line; this leads to a tour length bounded by $L^2$ (see Figure 2).

\begin{figure}[htbp]
  \centerline{\includegraphics[scale=0.75]{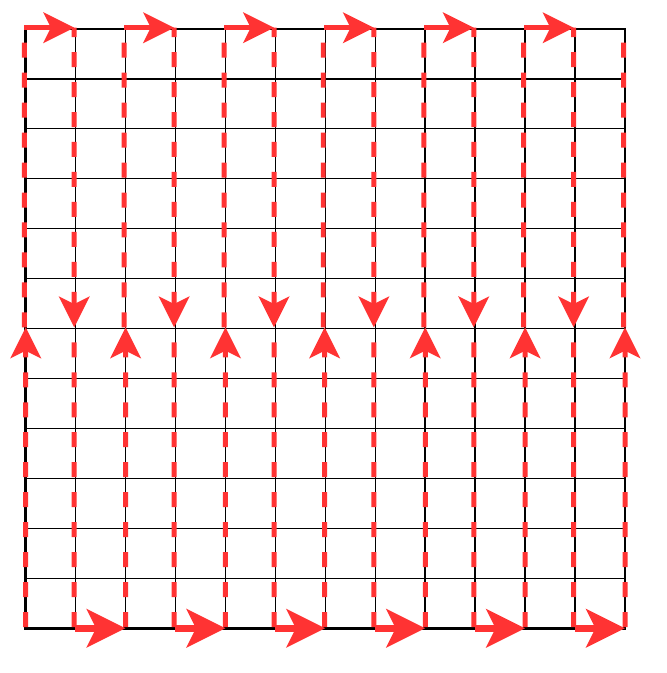}}
  \caption{The single salesman tour that leads to our upper bound on the TSP tour length}
  \label{fig:upper_bound}
\end{figure}

\subsubsection{Configurations} \label{subsubsec:configurations}

To describe our configuration, we first introduce some notation. Let $\{ \cdots \}$ denote a multiset, let $\ang{\cdots}$ denote a list, and let $(\cdots)$ denote an ordered tuple. Also, given positive integers $f$ and $g$, define
\[
    \{p\}_{f,g} = \{p_{f,1}, p_{f,2}, \ldots, p_{f,g}\} \qquad\text{and}\qquad
    \ang{(s,t)_{f,g}} = \ang{(s_{f,1},t_{f,1}), (s_{f,2},t_{f,2}), \ldots, (s_{f,g},t_{f,g})}.
\]
For a node in the quadtree, all possible groupings of the $k$ rounded paths are represented as a configuration $(A, B, C)$ where: 

\begin{flalign*}\label{eq:configuration}
A & ~ = ~ \big( \{p\}_{1,2 e_1}, \{p\}_{2,2 e_2}, \ldots, \{p\}_{k,2 e_k} \big),\\
B & ~ = ~ \big( \ang{(s,t)_{1,e_1}}, \ang{(s,t)_{2,e_2}}, \ldots, \ang{(s,t)_{k,e_k}} \big), \\
C & ~ = ~ \big(l_1, l_2, \ldots l_k\big)
\end{flalign*}
where 

\begin{itemize}
  \item $A$ is a $k$-element ordered tuple of multisets containing portals. The $h$th multiset is limited to no more than $r$ portals on each of the square's four edges. The total size of each multiset should be an even number $2e_h \leq 4r$.
  \item $B$ is a $k$-element ordered tuple of lists, the $h$th list representing all the entry-exit pairings for the $h$th tour. Each tuple $(s_{h,j},t_{h,j})$ represents a pairing of two distinct portals from the $h$th multiset in $A$.
  \item $C$ is an ordered tuple of $k$ lengths indicating that the $h$th path within this configuration has a rounded length of $l_h$.
\end{itemize}

The DP lookup table is indexed by quadtree node and configuration. A particular value in the table is set to true when the algorithm finds that the specified configuration is achievable within the node. Otherwise, the value is by default set to false. 

\subsubsection{Single Node Runtime} \label{subsubsec:single_node_runtime}

Throughout, we use ``$\log$'' and ``$\ln$'' to denote the base-2 and natural logarithms, respectively. We will make use of the following easy bounds on $\ln (1+x)$, which follows directly from the well known inequality $1 - \frac{1}{x} \leq \ln x \leq x-1$ for all $x \geq 0$.

\begin{lemma} \label{lem:natural_log_estimate}
For $0 < x < 1$, $\frac{x}{2} < \ln (1+x) < x$.
\end{lemma}

\begin{lemma} \label{lem:runtime-lem}
The number of rounded values at a level $i$ node is $O \left( \frac{1}{\eps} \log^2{\left( \frac{n}{\eps} \right) } \right)$.
\end{lemma}

\begin{proof}
An upper bound on the number of rounded values can be obtained by calculating the smallest integer $z$ where
$\delta(1 + \alpha)^z = L^2$. We take the $\log$ of both sides and solve for $z$:
\begin{align*} 
    \delta(1 + \alpha)^z 
        & ~ = ~ L^2 \\ 
    \frac{L}{2^i m} (1 + \alpha)^z
        & ~ = ~ L^2 \\
    (1 + \alpha)^z
        & ~ = ~ 2^i L m \\
    z 
        & ~ = ~ \frac{\ln \left(2^i L m \right)}{\ln \left( 1 + \alpha \right)}. 
\end{align*}
We can assume that $\eps < 1$, which means $\alpha = \frac{\eps}{2\log(L)} < \frac{1}{2\log(64 k n)} \leq \frac{1}{10}$. Lemma 2.1 can then be applied to $\ln(1 + \alpha)$ so that we have
\[
    z 
        ~ = ~ \frac{\ln \left(2^i L m \right)}{\ln \left( 1 + \alpha\right)} 
        ~ < ~ \frac{2 \ln \left( 2^i L m \right)}{\alpha}.
\]
Moreover, since $i \leq \log{L}$, we have $2^i \leq 2^{\log{L}} = L$, and thus
\[
    \frac{2 \ln \left( 2^i L m \right)}{\alpha} 
        ~ \leq ~ \frac{2}{\alpha} \ln \left( L^2 m \right).
\]
We finally plug in values for $L = \frac{64 k n}{\eps}$, $m = O\left(\frac{k\log{L}}{\eps} \right)$, and $\alpha = \frac{\eps}{2\log(L)}$, yielding
\begin{align*} 
    \frac{2}{\alpha} \ln{\left( L^2 m \right) } 
        & ~ = ~ \frac{4\log{L}}{\eps} \ln{\left( L^2 m \right) } \\
        & ~ = ~ \frac{4}{\eps} \log{\left( \frac{64 k n}{\eps} \right)} \ln{\left( \left( \frac{64 k n}{\eps} \right)^2 O \left(\frac{k \log{\frac{k n}{\eps}}}{\eps} \right) \right) } \\
        & ~ = ~ \frac{4}{\eps} \log{\left( \frac{64 k n}{\eps} \right)} \frac{1}{\log{e}} \left[ \log{ \left( \frac{64 k n}{\eps} \right)^2 + \log \left( O \left( \log{\frac{k n}{\eps}} \right) \right) + \log{\left( O \left( \frac{k}{ \eps} \right) \right)} } \right] \\
        & ~ = ~ O \left( \frac{1}{\eps} \left[ \log^2{\left( \frac{k n}{\eps} \right) }  + \log{\left( \frac{k n}{\eps} \right) } \log{\left(\frac{k}{\eps}\right)} \right] \right).
\end{align*}
We can assume by nature of the problem that the number of tours is bounded by the number of points. Therefore, since $k \leq n$, we simplify our Big O bound to
\[
    O \left( \frac{1}{\eps} \log^2 \left( \frac{n}{\eps} \right) \right).
\]

\end{proof}

\subsection{DP Recursion} \label{subsec:dp_recursion}

We define the $(m,r)$-multipath-multitour problem in the same way as Arora, but with the addition of our multiple tours. That is, an instance of this problem is specified by the following:

\begin{enumerate}[label=(\alph*)]
  \item[\stepcounter{enumi}\theenumi] A nonempty square in the shifted quadtree.
  \item For our $h$th tour, a multiset $\{p\}_{h,2 e_h}$ containing $\leq r$ portals on each of the square's four sides. The total amount of such portals in the $h$th multiset will be $2e_h \leq 4r$.
  \item A list of tuples $\ang{(s,t)_{h,e_h}}$ indicating ordered pairings between the $2e_h$ portals specified in (b).
  \item An instance of (b) and (c) for each of our $k$ tours.
\end{enumerate}

As shown earlier (in Section~\ref{subsubsec:configurations}), a configuration represents one possible combination of rounded tour lengths for a particular instance of the $(m,r)$-multipath-multitour problem.

\subsubsection{Base Case} \label{subsubsec:base_case}

Given a fixed choice in (1) the $2e_h$ portals for each of our $k$ multisets, containing at most $r$ portals from each square edge, and (2) the $\leq e_h$ entry-exit pairings for each of the $k$ tours, we have two cases for each leaf in our quadtree.

\begin{enumerate}
  \item There is no point to visit in the leaf: then we simply take the lengths of the given entry-exit pairing paths for each tour, setting all configurations that they round up to as true in the DP table. This takes $O(1)$.
  \item There is one point in the leaf: then we iterate over all assignments of the point to each tour within the node \emph{and} each entry-exit pairing within that tour (see Figure \ref{fig:leaves}). The path between the chosen entry-exit pairing within that tour is bent to visit the point, and we set that particular configuration of tours to true. We have $k$ tours and $\leq 2r$ pairings to choose for each tour, so this takes $2kr = O(kr)$.
\end{enumerate}

\begin{figure}[htbp]
  \centerline{\includegraphics[scale=1]{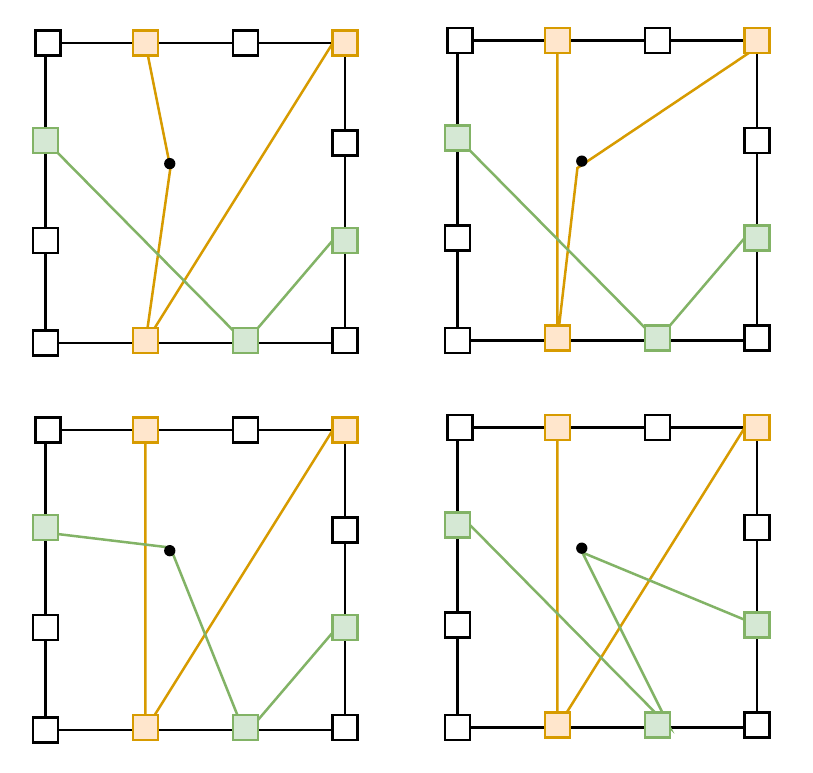}}
  \caption{A visualization of case (2) in our base case with $m=2$. The top row shows our two choices in which entry-exit path bends to visit the point when the orange tour is picked. The bottom row represents the possible configurations when the green tour is picked.}
  \label{fig:leaves}
\end{figure}

\subsubsection{Root Case} \label{subsubsec:root_case}

We consider the case where we have one depot, chosen from our $n$ locations, assigned for all of our $k$ tours. Then we require that each tour visits the depot within its leaf node. This involves $2r$ choices for each tour in deciding which entry-exit path bends to the point, and thus $O(r^k)$ different combinations of tour paths within the leaf. 

\subsubsection{General Case} \label{subsubsec:general_case}

Assume by induction that for level depth $> i$, the $(m,r)$-multipath-multitour problem has already been solved (that is, all possible configurations for these squares are stored in the lookup table). Let $S$ be any square at depth $i$, with children $S_1, S_2, S_3,$ and $S_4$. We have 

\begin{itemize}
  \item $(m+4)^{4rk}$ choices in (a), our multisets containing $\leq r$ portals on each of the square's four sides, for each of our tours
  \item $(4r)!^k$ choices in (b), the associated portal pairings for our k tours
\end{itemize}

For each of these choices in (a) and (b) for the outer edges of $S$, we have another set of choices within the four inner edges created by $S$'s children. Specifically, we have

\begin{itemize}
  \item $(m+4)^{4rk}$ choices in (a'), the multisets of $\leq r$ portals for each of our tours within the inner edges of $S_1, S_2, S_3, S_4$ 
  \item $(4r)!^k (4r)^{4rk}$ choices in (b'), the portal pairings for each tour within the inner edges (the term $(4r)^{4rk}$ represents the number of ways of placing the inner edge portal chosen in (a') within the entry-exit ordering of $S$'s tour. There are $\leq 4r$ outer edge portals to choose from, and the inner edge portal can only be crossed $\leq 4r$ times per tour).
\end{itemize}

\begin{figure}[htbp]
  \centerline{\includegraphics[scale=1]{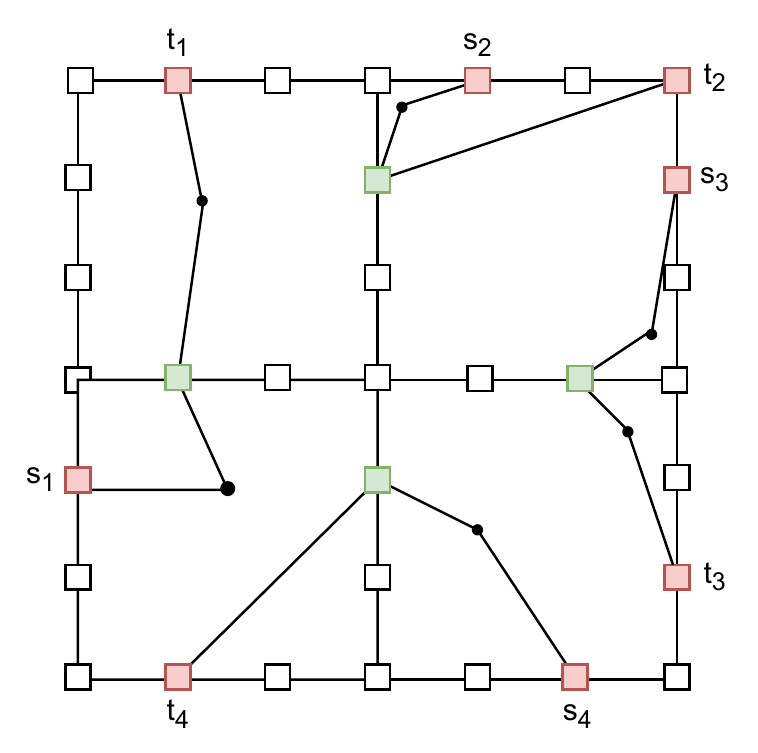}}
  \caption{An example of the recursion step for one of our $k$ tours. The red portals represent our fixed choice in (a) and (b), while the green display a particular configuration of our inner edge portals leading to a valid tour. The extension to $k$ tours works similarly, except we iterate through all rounded lengths for each fixed order of portals.}
  \label{fig:dp_recursion}
\end{figure}

Combining our choices in (a) and (b) for the outer edges with our choices in (a') and (b') for the inner edges leads to an $(m,r)$-multipath-multitour problem in the four children, whose solutions (by induction) exist in the lookup table with corresponding rounded lengths. There are up to $O \left( \left( \frac{1}{\eps} \log^2{\left( \frac{n}{\eps} \right) } \right)^{4k} \right)$ different rounded lengths stored for each of the $k$ tours, for each child square (Lemma~\ref{lem:runtime-lem}). We add each combination of rounded tour lengths together for the four children, and further round this value on level $i$'s scale. Finally, we set the resulting configurations of portal multi-sets, pairings, and tour lengths to true in our DP table.

Our runtime, then, is expressed as the product between the number of quadtree nodes ($T = O(n\log{L)}$) and the total number of choices outlined above:
\[
    O\left( T (m+4)^{8rk} (4r)^{4rk} (4r)!^{2k} \left( \left( 1/\eps \right) \log^2{\left( n/\eps \right) } \right)^{4k} \right) 
        ~ = ~ O \left(n \left( \left( k/\eps \right) \log {\left( n/\eps \right) } \right)^{O\left(k^2/\eps \right)} \right)
\]
Since we are treating $k$ as a constant, then, we simplify this bound to 
\[
    O \left(n \left( (1/\eps) \log {\left( n/\eps \right) } \right) ^ {O(1/\eps)} \right).
\]
It is evident that the above runtime of our DP algorithm is polynomial in $n$.

\section{Error Analysis} \label{sec:structure_thm}

We have three different sources of error in our algorithm. The first derives from our enforcement of $(m,r)$-light lines within the quadtree, as this requires us to shift each tour so they satisfy the portal pairing requirement. Moreover, our rounding incurs a small amount of cost on the min-max tour length. Finally, the perturbation step adds error to the solution as well. The following theorems will prove what those costs are, and we will show that compounding them on top of each other in our algorithm still leads to a min-max length bounded by $(1 + \eps) \OPT$.

\subsection{Arora's Error} \label{subsec:arora_structure}

\begin{lemma} \label{lem:structure_arora}
The total error incurred on the makespan when enforcing the tours to be $(m,r)$-light is $\leq (1 + \eps') \OPT$. 
\end{lemma}

\begin{proof}
Let there be a minimum internode distance of 8, with $(a, b)$ shifts applied randomly to the quadtree (see Arora). Then for $\eps' > 0$, Arora proves that the expected cost of enforcing one tour to be $(m,r)$-light is no more than $\frac{6g l}{s}$, where $l$ is the optimal length of the tour, $g = 6$ and $s = 12g/\eps'$. 

Take $l_{max}$ as the length of the optimal solution's makespan. Unlike Arora's problem, we have multiple tours. Thus we need to make sure that each one has a small enough expected cost so its length does not overpower $(1 + \eps') l_{max}$. We take $s = \frac{12gk}{\eps'}$ (where $m \geq 2s\log{L}$ and $r = s + 4$). This means that $\mathbb{E}(X_h) \leq \frac{\eps' l_i}{2k}$, where $X_h$ is a random variable representing the $h$th tour's incurred cost and $l_h$ is that tour's optimal solution length.

We know $P\left(\bigcap\limits_{h=1}^{k}{(X_h < \eps' l_h)}\right) = 1 - P\left(\bigcup\limits_{h=1}^{k}{(X_h \geq \eps' l_h)}\right)$. By Boole's Inequality,
\[
    P\left( \bigcup\limits_{h=1}^{k}{(X_h \geq \eps' l_h)}\right) 
        ~ \leq ~ \sum_{h=1}^k P(X_h \geq \eps' l_h).
\]
By Markov's Inequality, we know $P(X_h \geq \eps' l_h) \leq \frac{\mathbf{E}(X_h)}{\eps' l_h} = \frac{1}{2k}$, and therefore we have 
\[
    P\left(\bigcap\limits_{h=1}^{k}{(X_h < \eps' l_h)}\right) 
        ~ \geq ~ 1 - k\left(\frac{1}{2k}\right) 
        ~ =    ~ \frac{1}{2}.
\]
We can derandomize this result by iterating through each of the possible $(a,b)$ shifts and taking the lowest possible makespan, which runs in time $O(n^2)$.
\end{proof}

\subsection{Rounding Argument Error} \label{subsec:rounding_error}

\begin{lemma} \label{lem:error}
The error accumulated to the root level quadtree node by our rounding argument is $(1 + \alpha)^{\log{L}} \OPT$.
\end{lemma}

\begin{proof}
We take a look at each tour separately. At our root level, the total length of the rounded tour can be represented as $l_1^{(1)}$. But $l_1^{(1)}$ is just the sum of rounded lengths corresponding to the tour in each of its children. That is, $l_1^{(1)} = (1 + \alpha) \left(l_1^{(2)} + l_2^{(2)} + l_3^{(2)} + l_4^{(2)}\right)$. Our notation indicates that $l_1^{(2)}$ is the first indexed square at level 2, $l_2^{(2)}$ is the second, and so on. In general we have 
\[
    l_j^{(i)} 
        ~ = ~ (1 + \alpha) \left(l_{4j-3}^{(i+1)} + l_{4j-2}^{(i+1)} + l_{4j-1}^{(i+1)}  + l_{4j}^{(i+1)} \right).
\]
If we recursively expand each term in our initial equation, then, our root level rounded tour length becomes
\[
    l_1^{(1)} 
        ~ = ~ (1 + \alpha)^{\log{L} - 1} \left(l_1^{(\log{L})} + l_2^{(\log{L})} + \cdots + l_{4^{\log{L}}}^{(\log{L})} \right),
\]
where $\log{L}$ is the number of levels we have in our quadtree.

But $l_j^{(\log{L})}$ just represents the rounded lengths of our leaves, which are simply bounded by $\left( 1 + \alpha\right) t_j^{(\log{L})}$ (where $t$ represents the true optimum tour length in that node). Thus we have
\begin{align*}
    l_1^{(1)} 
        & ~ = ~ (1 + \alpha)^{\log{L}} \left(t_1^{(\log{L})} + t_2^{(\log{L})} + \cdots + t_{4^{\log{L}}}^{(\log{L})} \right) \\
        & ~ = ~ \left( 1 + \alpha \right)^{\log{L}} \OPT.
\end{align*}
\end{proof}

\begin{lemma} \label{lem:structure_rounding}
For $\alpha = \frac{\eps'}{2\log{L}}$, the total error of the minimum makespan when applying our rounding argument is at most $(1 + \eps') \OPT$. 
\end{lemma}

\begin{proof}
By Theorem 2, we know the error accumulated for each of the $k$ tours is bounded by $(1 + \alpha)^{\log{L}} \OPT$. Thus we need $(1 + \alpha)^{\log{L}} \leq (1 + \eps')$ to prove our error bound. We simplify our objective:
\[
    \log{(L)} \ln{(1 + \alpha)} 
        ~ \leq ~ \ln{(1 + \eps')}.
\]
By definition of the natural log, $\ln(1 + \alpha) \leq \alpha$ since $\alpha > 0.$ Thus we have
\[
    \log{(L)} \ln{(1 + \alpha)} 
        ~ \leq ~ \alpha \log{(L)}.
\]
Finally we plug in our $\alpha$ value:
\[
    \alpha \log{(L)} 
        ~ = ~ \frac{\eps'}{2\log{(L)}} \log{(L)} 
        ~ = ~ \frac{\eps'}{2}.
\]
By Lemma 2.1, we know $\ln(1 + \eps') \geq \frac{\eps'}{2}$. We therefore satisfy the inequality, as
\[
    \log{(L)} \ln{(1 + \alpha)} 
        ~ \leq ~ \frac{\eps'}{2} 
        ~ \leq ~ \ln{(1 + \eps')}.
\]
\end{proof}

\subsection{Total Error} \label{subsec:total_error}

We know the perturbation affects our optimum makespan length by adding at most $\frac{\eps}{4} \OPT$. Therefore we have a bounded makespan of $\left(1 + \frac{\eps}{4}\right) \OPT$ after this step. Given that $\eps' = \eps/4$, then, we layer the rounding and portal errors on top of this length to conclude that the total makespan length is bounded by
\[
    \left(1+\frac{\eps}{4}\right)^3 \OPT 
        ~ \leq ~ \left(1 + \eps\right)\OPT \qquad\big(\hbox{given that $\eps \leq \frac{\sqrt{13}-3}{2}$}\big),
\]
as desired.

\section{Problem Variations} \label{sec:variations}

\begin{enumerate}
  \item We require a subset of our $n$ points to be visited by a specific one of the tours. In this case, we would simply incur the same argument at the leaf nodes of these points as our base case in the DP Section~\ref{subsubsec:base_case}. This would just add a constant factor to the final runtime for each point in the subset.
\end{enumerate}

\section{Conclusion} \label{sec:conclusion}

We have demonstrated a PTAS for the min-max Euclidean multiple TSP. Our algorithm builds on top of Arora's PTAS for the single-tour TSP: that is, we restrict the search space of the solution by limiting tours to cross quadtree square edges a certain amount of times at evenly-spaced portals. On top of this simplification, we require that the $k$ resulting tour lengths within a quadtree node are rounded on a logarithmic scale depending on the square level. We then present a dynamic program that stores all possible restricted tour solutions within each level of the quadtree. This algorithm finds the optimal solution (within the portal and rounding simplifications) in polynomial time over $n$, with the degree growing as a function of $k$. Finally, we show that adjusting the optimal solution of an {\mtsp} instance to satisfy the portal and rounding requirements within the quadtree accumulates error of at most $(1 + \eps) \OPT$ units. 

Extending our PTAS solution for the {\mtsp} to higher dimensions would be a straightforward exercise. Moreover, generalizing the PTAS for the {\mtsp} to any shortest-path metric of a weighted planar graph should be an approachable problem following the solution presented by Arora, Grigni, Karger, Klein and Woloszyn for the single TSP~\cite{AGK98}. An open question remains as to whether there exists a PTAS that runs in polynomial time over both $n$ and $k$. Extensions such as variable speed and capacity restrictions on the $k$ salesmen also warrant exploration. Finally, we wonder whether similar DP techniques could be applied to the original problem that led us to the {\mtsp}: the Freeze Tag Problem~\cite{Bend06}, where robots in space aim to wake each other up in the least amount of time possible. 

\bibliographystyle{alpha} 
\bibliography{shortcuts,multi-tsp}

\newcommand{\etalchar}[1]{$^{#1}$}
\begin{thebibliography}{AGK{\etalchar{+}}98}

\bibitem[ABF{\etalchar{+}}06]{Bend06}
E.~Arkin, M.~Bender, S.~Fekete, J.~Mitchell, and M.~Skutella.
\newblock The freeze-tag problem: How to wake up a swarm of robots.
\newblock {\em Algorithmica}, 46:193--221, 2006.

\bibitem[AGK{\etalchar{+}}98]{AGK98}
S.~Arora, M.~Grigni, D.~Karger, P.~Klein, and A.~Wołoszyn.
\newblock A polynomial-time approximation scheme for weighted planar graph
  {TSP}.
\newblock In {\em Proc.\ Ninth Annu.\ ACM-SIAM Sympos.\ Discrete Algorithms},
  pages 33--41, 1998.

\bibitem[AHL06]{Ark06}
E.~M. Arkin, R.~Hassin, and A.~Levin.
\newblock Approximations for minimum and min-max vehicle routing problems.
\newblock {\em J.\ Algorithms}, 59:1--18, 2006.

\bibitem[AKTT97]{Asano97}
T.~Asano, N.~Katoh, H.~Tamaki, and T.~Tokuyama.
\newblock Covering points in the plane by $k$-tours: Towards a polynomial time
  approximation scheme for general $k$.
\newblock In {\em Proc.\ 29th Annu.\ ACM Sympos.\ Theory Comput.}, pages
  275--283. Association for Computing Machinery, 1997.

\bibitem[Aro98]{Arora98}
S.~Arora.
\newblock Polynomial time approximation schemes for {Euclidean} traveling
  salesman and other geometric problems.
\newblock {\em J. ACM}, 45:753--782, 1998.

\bibitem[BP19]{Be19}
A.~Becker and A.~Paul.
\newblock A framework for vehicle routing approximation schemes in trees.
\newblock In {\em Proc.\ 16th Internat.\ Sympos.\ Algorithms Data Struct.},
  pages 112--125, 2019.

\bibitem[CGS{\etalchar{+}}07]{Car07}
J.~Carlsson, D.~Ge, A.~Subramaniam, A.~Wu, and Y.~Ye.
\newblock Solving min-max multi-depot vehicle routing problem.
\newblock {\em J.\ Global Optim.}, 55, 1 2007.

\bibitem[EGK{\etalchar{+}}04]{Even04}
G.~Even, N.~Garg, J.~Könemann, R.~Ravi, and A.~Sinha.
\newblock Min–max tree covers of graphs.
\newblock {\em Oper.\ Res.\ Lett.}, 32:309--315, 2004.

\bibitem[FGLM95]{Franc95}
P.M. França, M.~Gendreau, G.~Laporte, and F.~M. Müller.
\newblock The $m$-traveling salesman problem with minmax objective.
\newblock {\em Transp. Sci.}, 29:267--275, 1995.

\bibitem[MN14]{Mats14}
T.~Matsuura and K.~Numata.
\newblock Solving min-max multiple traveling salesman problems by chaotic
  neural network.
\newblock In {\em Nonlinear theory appl. IEICE}, pages 237--240, 2014.

\bibitem[XR10]{Xu10}
Z.~Xu and B.~Rodrigues.
\newblock A 3/2-approximation algorithm for multiple depot multiple traveling
  salesman problem.
\newblock In {\em Proc.\ 12th Scand.\ Workshop Algorithm Theory}, pages
  127--138, 2010.

\end{thebibliography}
\end{document}